\documentclass[11pt]{article}

\usepackage[
backend=bibtex,
style=alphabetic,
sorting=ynt,
maxbibnames=19
]{biblatex}
\usepackage{palatino}

\bibliography{refs.bib}

\usepackage[margin=1in]{geometry}

\usepackage{amsmath,amssymb,amsthm}
\usepackage{mathtools}
\usepackage{float}

\usepackage{epsfig}
\usepackage{mdframed}

\usepackage{xr}
\usepackage{xr-hyper}

\usepackage[colorlinks = true]{hyperref}
\usepackage{xcolor}
\definecolor{darkred}  {rgb}{0.5,0,0}
\definecolor{darkblue} {rgb}{0,0,0.5}
\definecolor{darkgreen}{rgb}{0,0.5,0}
\hypersetup{
  urlcolor   = blue,         
  linkcolor  = darkblue,     
  citecolor  = darkgreen,    
  filecolor  = darkred       
}

\usepackage[utf8]{inputenc}

\usepackage{cleveref}
\usepackage{todonotes}
\newtheorem{theorem}{Theorem}
\newtheorem*{theorem*}{Theorem}
\newtheorem{proposition}[theorem]{Proposition}

\newtheorem{definition}[theorem]{Definition}
\newtheorem{algorithm}[theorem]{Algorithm}

\newtheorem{lemma}[theorem]{Lemma}

\def\Z{{\mathbb{Z}}}

\def\F{\mathbb{F}}

\def\poly{{\rm poly}}

\newcommand{\be}{\begin{eqnarray}}
\newcommand{\ee}{\end{eqnarray}}
\def\ba#1\ea{\begin{align}#1\end{align}}

\newcommand\ket[1]{{ |{#1} \rangle }}

\def\P{{\sf{P}}}

\newcommand{\vD}{{\vec{\Delta}}}

\newcommand{\vs}{{\vec{s}}}

\newcommand{\vx}{{\vec{\x}}}

\newcommand{\vy}{{\vec{\y}}}

\newcommand{\vz}{{\vec{\z}}}

\newcommand{\rl}{{r}}

\newcommand{\dman}{\Delta_{{\rm M}}}

\newcommand{\dimin}{{n}}

\renewcommand{\P}{\Pr}   

\newcommand{\eps}{{\varepsilon}}
\renewcommand{\epsilon}{\varepsilon}

\newcommand{\brs}{[\rl]}

\newcommand{\va}{{\vec{a}}}
\newcommand{\vc}{{\vec{c}}}

\newcommand{\vzero}{{\mathbf{0}}}

\newcommand{\matA}{\mathbf A}

\newcommand{\matG}{\mathbf G}

\newcommand{\pcs}[1]{\left| \boxed{\chi_\va} \cdot (#1) ,\va\right\rangle}

\renewcommand{\pcs}[1]{\ket{PC(#1),\va}}
\renewcommand{\pcs}[1]{\ket{PC(#1),\va}}
\renewcommand{\pcs}[1]{\ket{\psi_{#1}}}

\newcommand{\cube}[1]{\ket{C(#1)}}  

\newcommand{\cubep}[1]{\ket{C(#1)}}  
\newcommand{\cubet}[1]{\langle C(#1) }  

\renewcommand{\va}{\mathbf a}

\renewcommand{\vc}{\mathbf c}

\newcommand{\ve}{\mathbf e}

\renewcommand{\vs}{\mathbf s}
\newcommand{\vt}{\mathbf t}

\renewcommand{\vx}{\mathbf x}
\renewcommand{\vy}{\mathbf y}
\renewcommand{\vz}{\mathbf z}

\renewcommand{\vD}{\mathbf \Delta}

\renewcommand{\rl}{{\sigma}}

\begin{document}

\title{An Efficient Quantum Decoder for Prime-Power Fields}
\author{Lior Eldar}

\maketitle

\begin{abstract}
We consider a version of the nearest-codeword problem on finite fields $\F_q$ using the Manhattan distance, 
an analog of the Hamming metric for non-binary alphabets.
Similarly to other lattice related problems, this problem is NP-hard even up to constant factor approximation. 
We show, however, that for $q = p^m$ where $p$ is small relative to the code block-size $n$, there is a quantum algorithm that solves the problem in time $\poly(n)$, for approximation factor $1/n^2$, for any $p$.  
On the other hand, to the best of our knowledge, classical algorithms can efficiently solve the problem only for much smaller inverse polynomial factors.
Hence, the decoder provides an exponential improvement over classical algorithms, and places limitations on the cryptographic security of large-alphabet extensions of code-based cryptosystems like Classic McEliece.
\end{abstract}

\section{Introduction}

Error correcting codes are linear subspaces of finite-field vector spaces
that allow to protect information against random, and even adversarial errors.
The problem of designing good, efficiently decodable, error-correcting codes is notoriously difficult, and is in fact tantamount to an art form: interestingly, it is difficult not because large minimal-distance codes are hard to find (in fact typically a random code does have a large minimal distance), but rather because it is hard to find such codes that are simultaneously efficiently decodable. 

The Maximum Likelihood Decoding (MLD) Problem of error correcting codes is well known to be NP-hard since the work of Berlekamp, McEliece and Tilborg \cite{BMT78}. 
Formally, for the MLD problem we are given a "syndrome" $\vs \in \F_q^m$, a "parity check matrix" $\matA$ and are asked to find $\ve \in \F_q^n$ of weight at most $w$ such that $\matA \ve = \vs$.  In a related problem, called the Nearest Codeword Problem (NCP) \cite{ABSS97,  Reg03}, we are given a target vector $\vt$, a generator matrix $\matA$ and are asked to find the closest codeword to $\vt$, namely $\vs$ such that $\matA \vs$ is closest to $\vt$, provided that this distance is at most $w$.
This problem too, is known to be NP-hard  even to sub-polynomial approximation factors \cite{ABSS97} under reasonable complexity assumptions.

\subsection{Defining BNCP}

The NCP problem is analogous to the closest vector problem (CVP) defined on Euclidean lattices, which is also notoriously hard (see e.g. \cite{DKRS03}).  Yet, as is often the case in error-correcting scenarios where the error has bounded length, and similarly to the MLD problem, one can consider a bounded error variant of NCP which we call here the Bounded NCP, namely we  given a target vector $\vt$, a matrix $\matA$ we are asked to find $\vt$'s closest vector in the span of $\matA$, provided that this distance is, say, at most $1/10$ of the minimal error correcting distance.

\begin{definition}\label{def:bncp}

\textbf{Bounded Nearest Codeword Problem, Hamming Metric}

\noindent
Given an error correcting code ${\cal C} = [n,k,d] \subseteq \F_2^n$,
where $d$ is the minimal Hamming distance between any pair of distinct codewords in ${\cal C}$, and is
generated by matrix $\matA \in \F_2^{n \times k}$, and a vector $\vt$ such that for some $\vs \in \F_2^k$:
$$
\Delta(
\vt, \matA \vs) 
\leq 
\eps \cdot d
$$
where $\Delta(\vx,\vy)$ refers to the Hamming distance between $\vx,\vy$.
We are asked to find $\vs$.
\end{definition}

The definition above uses the Hamming distance between words - namely the number of positions in which two strings are different.  Yet, for non-binary $q$-ary alphabets, it is of interest to consider different metrics that take into account the actual labels. 
Considering the alphabet as the additive group $\Z_q$, one such distance is called the Lee Distance 
\footnote{In fact, the Lee distance over the ring of integers generates a metric space over the ring since it satisfies, in addition to positivity, and symmetry, the triangle inequality.}
and is defined as follows:
$$
\Delta_{\rm L}(\vx, \vy) = \sum_{i=1}^n
\min\{ |x_i - y_i|, q - |x_i - y_i| \}
$$
Another distance on $\Z_q$ is the well-known Manhattan distance corresponding to the $\ell_1$-norm of Euclidean space:
$$
\dman(\vx, \vy) 
= 
\sum_{i=1}^n
|x_i - y_i|
$$

One can then reconsider Definition \ref{def:bncp} for large alphabets:
\begin{definition}\label{def:bncp2}

\textbf{Bounded Nearest Codeword Problem ($\eps$-BNCP), Manhattan Distance}

\noindent
Given is an error correcting code ${\cal C} = [n,k,d]$,
where $d$ is the minimal Manhattan distance between any pair of distinct codewords in ${\cal C}$. ${\cal C}$ is
generated by matrix $\matA \in \F_q^{n \times k}$.
We are also given a vector $\vt$ such that for some $\vs \in \F_q^k$:
$$
\Delta_M( \vt, \matA \vs) \leq \eps \cdot d
$$
We are asked to find $\vs$.
\end{definition}

\subsection{Hardness of BNCP}

In general, BNCP has no known efficient classical algorithms, 
and the assumed hardness of this problem is, in fact, central to the security of the McEliece cryptosystem (\cite{CS98}, \cite{mceliece}), one of the finalists in the NIST effort to design quantum-secure cryptosystems \cite{NIST20}.

For $q=2$ the definition above generalizes the Hamming metric, hence BNCP is NP-hard to solve for general $q$.
Let us now examine the behavior of its complexity for specific values of $q$.
On one hand, for $q$ which is a prime number, $\F_q$ inherits its multiplication / addition table from $\Z_q$ and in that case the problem is nearly identical 
\footnote{Up to the fact that $q$-ary lattices are in fact integer lattices with a shortest vector length at most $q$, whereas $\F_q$ lattices can have longer shortest vectors, for example the $1$-dimensional code generated by $(1,2,\hdots, q-1)$}
to the Bounded Distance Decoding for $q$-ary lattices, a problem whose $\eps$-approximation is known to be at least as hard as computing the (unique) shortest vector of an integer lattice up to a factor $1/\eps$ \cite{LM09}.

As further testament for the generic hardness of this problem: the result of \cite{ABSS97} on the hardness of approximation of the decisional version of the nearest-codeword problem w.r.t. the Hamming metric, can be readily extended to the Manhattan  / Lee distances, albeit with a diminished promise gap:
\begin{theorem}\label{thm:hardness}

\textbf{NP-hardness of constant factor approximation of decisional BNCP}

\noindent
Let ${\cal C} = [n,k,d]$ be some code of $\F_q^n$.
There exists a constant $c>0$ such that
if
$q = p^m$ for some integer $m$, and $c>p$ then it is NP-hard to decide whether a vector $\vt \in \F_q^n$ is at Manhattan distance at most $L$ or at distance at least $(c/p) \cdot L$ from ${\cal C}$.
\end{theorem}
\noindent
The proof appears in the appendix.

The resemblance of large-alphabet BNCP to $q$-ary BDD is also apparent in the behavior of random lattices: one can check that just as random $q$-ary lattices have relatively long shortest vectors, similar bounds are satisfied w.r.t. the Manhattan/Lee distance for $\F_q$ random lattices, for any $q$.  We provide a formal statement w.r.t. the Manhattan distance.
\begin{lemma}\label{lem:random}
Let $A$ be a uniformly random $n\times k$ matrix over $\F_q$ where $n \geq k \log(q)$.
Then
$$
\P_A
\left(
\min_{\vx \neq \vy, \vx,\vy \in {\cal C}} \Delta_{M,q}(\vx, \vy)
\geq
q^{1-k/n}/2
\right)
\geq
1 - 2^{-n/200}
$$
\end{lemma}
\noindent
The proof appears in the appendix.
Hence for a "typical" error-correcting code the shortest vector can be of length, say, $q^{0.99}$ even for linear rate codes ($k/n = 0.01$). Thus, one can define non-trivial $\eps$-BNCP on random ensembles for very small values of $\eps$, say $\eps = 1/\sqrt{q}$.
When allowing $q$ to grow with $n$ we thus achieve a setting that is similar to lattice problems used for cryptography.  In this analogy: $\eps$ corresponds to the security parameter of the instance (usually signifed by $\alpha$ for the Learning-with-Errors cryptosystem), and an instance is considered to be "hard", or "secure", for cryptographic purposes whenever $\eps \cdot q > \sqrt{n}$.

\subsection{Main Results}

Despite the apparent difficulty of this problem, we show, that surprisingly,
for a code ${\cal C} = [n,k,d]$ over $\F_q$,
for $q$ which is a prime power, namely $q = p^m$ for prime $p$, there exists an efficient quantum algorithm that solves BNCP for $\eps<1/(p \cdot n^2)$:
\begin{theorem}(sketch of Theorem \ref{thm:main2})

\noindent
\textbf{A Quantum Decoder for Prime-Power Fields}

\noindent
There exists a quantum algorithm that for any $q = p^m$, where $p$ is prime, solves $\eps$-BNCP on $\F_q$ w.r.t. the Manhattan distance for $\eps < 1/(pn^2)$ in time $\poly(n,p, \log(q))$.
\footnote{The dependency on $m$ is accounted for in the dependency on $\log(q)$.
We note that this algorithm can be easily adapted to the Lee metric by symmetrizing over the difference from $q$.}
\end{theorem}
We note that by a slight assumption on the distance to the lattice being an integer power of $p$ one can increase $\eps$ to $1/n^2$.
We note that for $q = 2$ the above does not provide a meaningful statement since the largest possible value for the minimal error correcting distance is at most $n$.  It is only for relatively large values of $q$, say $q = n^3$ that this approximation provides a non-trivial statement.

\subsection{Classical Algorithms}

\subsubsection{Direct Inversion}

Consider an error correcting code
 ${\cal C} \subseteq \F_q^n$ for $q = p^m$ for some integer $m$, and recall that each element of $\F_q$ can be regarded as an $m$-tuple of numbers in $\F_p$.
Given a target vector $\vy \in \F_q^n$ that is close to ${\cal C}$,
on may be tempted to think that for a sufficiently small noise level each coordinate of $\vy$, viewed as an $m$ dimensional vector in $\F_p^m$, has sufficiently many noise-free coordinates - namely the "most-significant" bits, that allow us to determine $\vy$'s closest vector precisely.
In other words, instead of solving the optimization problem:
$$
\min_{\vx \in {\cal C}} \Delta(\matA \vx, \vy)
$$
we solve the linear system of equations:
$$
\tilde \vy = \tilde \matA \cdot \tilde \vx
$$
where $\tilde \vy, \tilde \vx$ correspond to the top-bits in the representations of $\vx, \vy$ as vectors in $\F_p^m$, and $\tilde \matA$ is the corresponding $\F_p$ sub-matrix of $\matA$.

However, such a scheme fails immediately the test of invertibility: one can easily generate a matrix $\matA \in \F_q^{n\times n}$ that is invertible over $\F_q$,
yet regarding $\matA$ as an $m n \times m n$ linear operator over $\F_p$ and taking the submatrix $\matA_k$ corresponding to the top $k<m$ coordinates in each tuple results in a matrix that fails to be                                                                                                                                                                                                                                                                                                                                                                                                                                                                                                                                                                                                                                                                                                                                                                    invertible over $\F_p$. 
For example, considering $q = 4 p = 2$, and constructing $\F_4$ via the irreducible polynomial $x^2 + x + 1$ over $\F_2$ one can check that the matrix 
$$
\matA = 
\begin{bmatrix}
1 & 2 \\
3 & 0
\end{bmatrix}
$$
over $\F_4$
can be written as a linear operator over $\F_2^{4}$ as follows:
$$
\tilde\matA 
= 
\begin{bmatrix}
1   &0 	&0 	&1 \\
0   &1	&1	&1 \\
1	&1	&0	&0 \\
1	&0	&0	&0
\end{bmatrix}
$$
However, extracting the sub-matrix corresponding the top coordinate of each vector results in the following matrix:
$$
\tilde\matA_1 
= 
\begin{bmatrix}
1    	&0 	 \\
1		&0	 \\
\end{bmatrix}
$$
which is not invertible over $\F_2$.

It is plausible to hope that such examples are pathological, in the sense that they rarely appear for random codes.  Yet, even for random codes this problem is prevalent.
Let us consider concrete estimates:
according to Lemma \ref{lem:random} the typical minimal distance $d$ of a random $\F_q$ code is at least $q^{1-k/n}/2$.
If 
$$
\Delta \eps  < q^{1-k/n}/2
$$ 
then 
viewing each number $x\in \F_q$ as an $m$-dimensional vector $\vx \in \F_p^m$ we have that
the top
$m k / n$ $\F_p$ coordinates of each coordinate of the target vector $\vy\in \F_q^n$ are equal to the corresponding coordinates of some codeword $\vc \in {\cal C}$.
Hence, as above, we can
write a linear system of equations over $\F_p$:
$$
\tilde \vy = \tilde \matA \cdot \tilde \vx
$$
corresponding to the top coordinates in the $\F_p$ expansion of each $\F_q$ number.
By assumption of the random instance, the submatrix $\tilde \matA$ is in fact a random $mk \times mk$ matrix over $\F_p$, which is invertible with probability roughly $p^{-p}$,
i.e. independently of $n$.

In order to increase the probability that $\tilde \matA$ is invertible over $\F_p$
to nearly $1$, one would need to decrease the parameter $\eps$ controlling the relative distance to the lattice further so that
$$
\Delta \eps \sim q^{1-\beta k/n} 
$$
for $\beta = \Omega(\log(p))$ but this implies increasing the promise from $\eps$ to roughly $\eps^{\beta-1}$.

\subsubsection{Information-Set Decoding}

Other
classical attacks against McEliece that might also be relevant for BNCP include mainly variants of the Information-Set Decoding algorithm (see e.g. \cite{PS10}, and in the context of the Lee metric see more recent works in  \cite{CDE21, HW19}), but that algorithm's run-time scales exponentially in the rate $k$.
Since "good" codes, i.e. codes which have linear distance and linear rate are usually the target codes considered for both theoretic and practical applications, such algorithms are prohibitive.

\subsubsection{Summary}

Thus, to the best of our knowledge, efficient decoding w.r.t. the Manhattan distance is only available for random ensembles where the submatrix $\tilde \matA$ is invertible with overwhelming probability.
We compute below the quantum-classical separation for 
$p = 16$, $q^{k/n} = n^6$, $d = q \cdot n^{-6}$ and hence
$q^{\beta k/n} = n^{24}$:
\begin{center}

\textbf{Quantum-Classical Separation for Worst-Case/Average-Case instances of BNCP for $q = 16^m$.}
	
  \begin{tabular}{ c | c | c }
    \hline
    Worst-case & Quantum & Classical \\ \hline
    $1/n^2$      & $\poly(n)$ & $e^{\Omega(n)}$ \\ \hline
    $1/n^{18}$ & $\poly(n)$ & $e^{\Omega(n)}$  \\ \hline
  \end{tabular}
  \quad
    \begin{tabular}{ c | c | c }
    \hline
    Average-case  & Quantum & Classical \\ \hline
    $1/n^2$      & $\poly(n)$ & $e^{\Omega(n)}$ \\ \hline
    $1/n^{18}$ & $\poly(n)$ & $\poly(n)$  \\ \hline
  \end{tabular}
\end{center}

\subsection{Context on the Main Result}

Hard computational problems related to lattices and error correcting codes have  resisted efficient quantum algorithms for nearly two decades now, despite their underlying Abelian structure that presumably makes them more susceptible to such algorithms.
This resistance has given rise to the belief that quantum computers cannot outperform classical ones on problems that require any form of "bounded distance decoding" even with an inverse polynomial promise gap.  
Our main result suggest that this intuition may be false.

Notably, our result, as it is, does not directly pose a threat to any known public key crypto-system since its parameter range is quite different than those considered for established PQC systems \cite{NIST20}: for the main code-based PQC cryptosystems the alphabet size is constant with the block-size (see e.g. BIKE, HQC, Classic McEliece), whereas for lattice-based cryptosystems, 
namely descendants of the LWE cryptosystem \cite{Reg09, Pei09, BLPRS13},
where the alphabet size is in fact allowed to grow with the lattice dimension, the underlying algebraic structure is not a finite field but rather the ring of integers.

Yet, we believe that this parameter mismatch does not capture the full story: our work here suggests that using low-order QFT's to optimize over high-order groups, 
(in this case, prime-power fields) 
in conjunction with the recent construction of \cite{EH22} of approximate eigenvectors
of the vector shift operator, does in fact lead to an exponential quantum speed-up for lattice related problems.
We hope that further study of the approach outlined here will lead to additional discoveries in this field, classical or quantum.

\section{Preliminaries}

\subsection{Notation}

$\F_q$ denotes the field of order $q$, $\omega_p$ denotes the $p$-th root of unity.
For $\vx, \vy \in \F_q^n$ we use $\Delta_{M,q}(\vx, \vy)$ to denote the Manhattan distance between $\vx, \vy$ (see Definition \ref{def:man1}).
An $\F_q$ error-correcting code is denoted by ${\cal C} = [n,k,d] \subseteq \F_q^n$ where $n$ is the block-length, $k$ is the rate, and $d$ is the minimal {\it Manhattan} distance between any pair of codewords:
$$
d = \min_{\vx\neq \vy, \vx, \vy \in {\cal C}}
\Delta_{M,q}(\vx, \vy).
$$ 
A code ${\cal C}$ of rate $k$ is generated by a matrix $\matA \in \F_q^{n\times k}$.
Often we omit the subscript $q$ when it is clear from context.
Similarly, we use $\Delta_{\rm L}$ to signify the Lee distance.
For a subset $S \subseteq \F_q^n$
$$
\Delta(\vx, S)
$$
signifies the minimal distance between $\vx$ and any $\vy \in S$.
For $\vx \in \F_q^n$ let $U_\vx$ denote the shift operator:  $U_{\vx} \ket{\vy} = \ket{\vx+\vy}$.
The quantum Fourier Transform on the $n$-dimensional vector space (module) w.r.t. the ring of integers $\Z_p$ is denoted by ${\cal F}_p^n$.

\subsection{Vector Representation of Prime Power Fields}

Let $p$ be a prime number, and 
let $\F_q$ be a prime number field of order $q = p^m$ relative to some degree-$m$ irreducible polynomial $P \in \F_p[x]$:
$$
\F_q = F_p[x] / P
$$

For $a \in \F_q$, $q = p^m$, let $\hat{a} \in \F_{p}^m$ denote its $\F_p$ vector representation. Likewise for a vector $\va \in \F_q^k$ let $\hat{\va} \in \F_p^{m \cdot k}$ denote the concatenation of the $\F_p$ expansion of each of its coordinates. 
As an additive group, $\F_q$ is equal to the $m$-dimensional vector space over $\F_p$:
$$
\F_q = \F_p \times \F_p \times \hdots \times \F_p
$$
We assign $q$-ary labels $a\in \F_q$ to the elements of $\F_p^m$ as a $p$-ary expansion order:
\begin{align}
\widehat{\vzero} &=     
\underbrace{(0, 0, \hdots, 0, 0)}_{m \mbox{ coordinates}}
\nonumber \\ 
\widehat{\mathbf{1}} &= (0,0, \hdots, 0, 1) \nonumber \\
& \vdots \nonumber \\
\widehat{\mathbf{p-1}} &= (0,0, \hdots, 0, p-1) \nonumber \\
&\vdots \nonumber \\
\widehat{\mathbf{q-p}} &= (p-1,p-1,\hdots, p-1, 0) \nonumber \\
\widehat{\mathbf{q-p+1}} &= (p-1,p-1,\hdots, p-1, 1) \nonumber \\
& \vdots \nonumber \\
\widehat{\mathbf{q-1}} &= (p-1,p-1, \hdots, p-1, p-1) \nonumber
\end{align}
This corresponds to interpreting $\hat{a}$ as the $\F_p$ coefficient vector of the polynomial corresponding to $\hat{a}$:
$$
\hat{a}(x) = \sum_{i=1}^m \hat{a}_i \cdot x^{i-1}
$$
When we use the ordering $x<y$ on $x,y\in \F_q$ it means that $x<y$ as numbers in $\Z$.
If $\sigma = p^r$ and $x<\sigma$ we will often use the notation 
$$
\hat{x} \in 0^{n-r} [p]^r
$$
signifying that $x$'s representation as a $p$-ary vector has $0$ in the first $n-r$ positions (MSB).

Unless stated otherwise, for $x,y \in \F_q$ the expressions $x + y, x \cdot y$ denote addition / multiplication over $\F_q$.
The following proposition is immediately implied by definition:
\begin{proposition}\label{prop:1}

$$
\forall \vx, \vy \in \F_q^n
\quad
\widehat{\vx + \vy}
=
\hat{\vx} + \hat{\vy}
$$
where the addition in LHS is over $\F_q$ and the RHS addition is over $\F_{p}^m$.
\end{proposition}
For example for $1\in \F_q$ and $3\in \F_q$, $q = 16 = 2^4$ we have that $1 + 3 = 2$.

\subsection{Extending The Manhattan Distance to Prime Power Fields}

The Manhattan distance was developed as an alternative to the Hamming distance for transmission of non-binary signals taken from some $q$-ary alphabet. 
The Manhattan distance 
$$
\Z_q^n \times \Z_q^n \to \Z^+
$$ 
is defined as follows:
$$
\forall \vx, \vy \in \Z_q^n,
\quad
\dman(\vx, \vy)
:=
\sum_{i=1}^n
|x_i - y_i|
$$

In the context of linear codes, one considers the finite field $\F_q$ and not $\Z_q$.  These objects are quite different, and are equal only for prime $q$, yet in this study we consider $q$ which itself is a prime power, i.e. $q = p^m$.
To this end we define a mapping from $\F_q$ to the ring of integers $\Z_q$ using the natural $p$-ary expansion above:
\begin{definition}

\textbf{$p$-ary expansion mapping}

\noindent
For $a\in \F_q$, $q = p^m$ we define $\hat{a} \in \Z_q$ by writing $a$ as a vector $(\hat{a}_1,\hdots, \hat{a}_m) \in \F_p^m$, i.e. $\hat{a}_i \in \F_p$ and defining the polynomial
$$
\hat{a}(x)
=
\sum_{i=1}^{m} \hat{a}_i \cdot x^{i-1}
$$
we then set the $\Z_q$ representation of $a$, namely $\tilde{a}$, as the evaluation of the polynomial $\hat{a}(x)$ at point $x = p$:
$$
\tilde{a} = \hat{a}(p) = \sum_{i=1}^m \hat{a}_i p^{i-1}
$$
For $\vx \in \F_q^n$ we define $\tilde{\vx} \in \Z_q^n$ as applying the map above coordinate-wise.
\end{definition}
We then extend the Manhattan distance to finite fields $\F_q$ by setting:
\begin{definition}\label{def:man1}

\textbf{Manhattan Distance for $\F_q$:}

\noindent
$$
\forall \vx, \vy \in \F_q^n
\quad
\Delta_{M,q}(\vx, \vy)
=
\Delta_{M}(\tilde{\vx}, \tilde{\vy}).
$$
and define the length of a vector $\vx \in \F_q^n$ as its distance from $0$:
$$
\| \vx \|_M = \Delta_{M,q}(\vx, 0)
=
\Delta_M(\tilde\vx, 0)
$$
\end{definition}

When we consider $\F_q$-linear codes we would like to consider the "minimal distance" of a code, or the distance of a given word from an $\F_q$ codespace, however, 
since the Manhattan distance is not a metric, then in particular
it is not shift invariant on $\Z_q$: for example, setting $q=p$ for prime $p$ and $n=1$ we have:
$$
p-1
=
\Delta_M(p-1,0) 
\neq
\Delta_M(p-1+1,0+1)
=
\Delta_M(0,1)
=1
$$
Similarly, the shift-invariance property does not hold for the distance ${\Delta}_{M,q}$.
However what one can show, is that specifically for the $p$-ary expansion mapping the following gap-presevation property does hold:
\begin{proposition}\label{prop:gap}

\textbf{Gap Preserving Property}

\noindent
For all $\vx, \vy \in \F_q^n$
if for some $r<m$
$$
\Delta_{M,q}(\vx, \vy) < p^r
$$
then
$$
\Delta_{M,q}(\vx- \vy,0)
\leq
n \cdot p^r
$$
In particular, if ${\cal C}\subseteq \F_q^n$
is such that for some $r<m$
$$
d = \min_{\vx\neq \vy, \vx, \vy \in {\cal C}}
\Delta_{M,q}(\vx, \vy) < p^r
$$
then the shortest vector $\vx$ of ${\cal C}$
satisfies:
$$
\| \vx\|_M <  n \cdot p^r
$$
\end{proposition}

\begin{proof}
If $\Delta_{M,q}(\vx, \vy) < p^r$ then 
$$
\forall i\in [n]
\quad
\Delta_M(\widetilde{\vx_i}, \widetilde{\vy_i}) < p^r
$$
This implies that for each $i\in [n]$ the respective $p$-ary expansions of $\widehat{\vx_i}, \widehat{\vy_i} \in \F_p^m$ are identical on indices $m,m-1,\hdots, r+1$.
Thus $\vz = \vx - \vy$ (subtraction over $\F_q$) is such that for each $i \in [n]$ the $p$-ary expansion of $\vz_i$, i.e. $\widehat{\vz_i}$ is $0$ for the top $m-r$ coordinates.
In particular $\widetilde{\vz_i} < p^r$ so 
$$
\| \vz \|_M 
\equiv
\sum_{i\in [n]} \widetilde{\vz_i}
< n \cdot p^r
$$

\end{proof}

We note that if one replaces the Manhattan distance with the Lee distance, which is in fact a metric on $\Z_q$ one obtains a metric space on $\F_q^n$ via the $p$-ary expansion mapping defined above, and under the partial order on elements of $\F_q$ defined above for the $p$-ary expansion.
This would imply, in particular that one would be able to improve the performance of the proposed quantum algorithm by a factor of $n$.
Still, we decided to develop this study using the Manhattan distance and not the Lee distance, losing the property of a metric space, for a more general statement.

\subsection{Invertibility of Random Matrices over Finite Fields}

The well-established theory of random matrices over finite fields characterizes the probability that a uniformly random matrix over a finite field is invertible as follows:
\begin{lemma}\label{lem:inv}

\textbf{Theorem 1.1 in \cite{M10}}

\noindent
Let $A\sim U[\F_p^{n \times k}]$.
There exists a constant $c>0$ such that:
$$
\P_{\matA}
\left(
\matA 
\mbox{ is invertible}
\right)
\geq
\prod_{k=1}^{\infty}
(1 - p^{-k}) 
- e^{-c T}
$$
\end{lemma}

\section{Quantum PCS States on Finite Fields}

In \cite{EH22} the authors define the "Phased Coset State" (or PCS) on $q$-ary lattices as a certain superposition on the lattice, comprised of copies of a bounded function - each centered around an individual lattice point and multiplied by a phase that depends on that lattice point.
Here we redefine the PCS on finite fields:

\begin{definition}\label{def:pcs}

\noindent
\textbf{PCS on Finite Fields}

\noindent
For $\sigma \in \F_q$ define the set $\brs =\{0,\hdots, \sigma-1\} \subseteq \F_q$ and $\brs^n$ as the n-th fold product thereof.
  \begin{enumerate}
    \item Define the {\em cube state}
 anchored at a point $\vy \in \F_q^n$ by
 \[\cube{\vy} =
 \sigma^{-n/2} \cdot
   \sum_{\vz \in \brs^\dimin} \ket{\vy +\vz }.\]
\item 
Let ${\cal C} = [n,k,d] \subseteq \F_q^n$ with $q = p^m$.
The {\em phased cube state} with label $\widehat{\va} \in \F_p^{m \cdot k}$ is the following state:
 \[
   \pcs{\widehat{\va}}
 = 
 q^{-k/2}
 \cdot
 \sum_{\vc \in \F_q^k} 
 \omega_p^{\hat{\va} \cdot \hat{\vc}}
 \cube{\matG\vc }
 = 
 q^{-k/2} \cdot \sigma^{-n/2} \cdot
 \sum_{\vc \in \F_q^k} 
 \omega_p^{\hat{\va} \cdot \hat{\vc}}
\sum_{\vz\in \brs^\dimin} \ket{\matG\vc +\vz}.
\]
\end{enumerate}
\end{definition}
Note that the quantum state is defined on a register with numbers in $\F_q$ whereas the phase that multiplies each basis element is a power of the primitive root $\omega_p$.

\begin{lemma}[Cube state properties] \ \label{lem:cs-props}
  \begin{enumerate}

    \item\label{it:shift1}
      $\forall \vx,\vy \in \F_q^n, U_\vx \cube{\vy} = \cubep{\vx+ \vy}$,
      and the transformation $\cube{\vy}\ket{\vx}$ to
      $\cubep{\vx+\vy}\ket{\vx}$ is computable in time $\poly(n, \log q)$.

  \item Let $\cube{\vy}$ be a cube state of side length $\sigma = p^r$ for some $r>0, \vy \in \F_q^n$ and
    let $\vD\in \F_q^n$.
    \begin{enumerate}

    \item \label{it:dist-close}  
 If $\|\vD\|_{M,q} \leq \sigma$ then
 $\cube{\vy} = \cubep{\vy+\vD}$.
        
    \item \label{it:dist-far} If  $\|\vD\|_{M,q} > n \cdot \sigma$ then $\cubet{\vy}  \cubep{\vy+\vD}  =0$.
    \end{enumerate}
      
 \end{enumerate}
\end{lemma}

Consider the implication of Item \ref{it:dist-close}: it implies that the PCS is not "periodic" on the code-space in the usual sense of having  symmetric support around each codeword (/lattice point).  
The function of symmetric support around each codeword is a different function which is the convolution of the Hamming ball and the code-space.
Rather, the support of the cube-shaped super-position starts at the point which is the original codeword with all the right-most $r$ coordinates erased.
Note that the "erased" information is  encoded in the phase that multiplies each cube.
For example, the cube anchored at a codeword $c$ such that $\hat{c} \in [p]^{n-r} 0^{r}$ is situated to the "bottom-right" of the codeword, whereas if $\hat{c} \in [p]^{n-r} 1^{r}$ it is situated on the "top-left" of the codeword.

\begin{proof}

\item{\textbf{Item:} \ref{it:shift1}}

\noindent
\begin{align}
\forall \vx \in \F_q^n \quad
U_\vx \cube{\vy} 
&= \sigma^{-n/2} 
\cdot
\sum_{\vz \in \brs^n}
U_\vx \ket{\vy+\vz } \\
&= \sigma^{-n/2}
\cdot
\sum_{\vz \in \brs^n} \ket{\vx +\vy+\vz} \\
&= \cubep{\vx+ \vy}
\end{align}
Therefore, given $\cube{\vy}\ket{\vx}$, one
addition from the
second register into the first register results in $\cubep{\vx+\vy}\ket{\vx}$.

\item{\textbf{Item:} \ref{it:dist-close}}

\noindent
Start with
$$
\Big\| 
      \cube{\vy}-\cubep{\vy + \vD} \Big\|^2 =
2 \cdot (1 - \Re(\cubet{\vy}\cubep{\vy + \vD} )).
$$
We have:
\begin{align}\label{eq:prod}
& \cubet{\vy} \cubep{\vy+\vD } 
= \cubet{\vzero}\cube{\vD}
        &\text{Item \mbox{\ref{it:shift1}} for the shift by $\vy$}
\end{align}
Since $\|\vD\|_{M} \leq \sigma$ then
$$
\forall i\in [n]
\quad
\widetilde{\vD_i} \leq \sigma.
$$
Hence, for each $i$ the set $\brs$ is invariant under shift by $\vD_i$:
\begin{align}
\vD_i + \brs 
&= 
\{ x + \vD_i, \ x\in \brs\} \\
&=
\{ x + \vD_i, \ \hat{x} \in 0^{n-r} [p]^r \} 
\quad
\mbox{By the assumption that $\sigma = p^r$}
\\
&=
\{ x + \vD_i, \ \widehat{x + \vD_i} + \widehat{\vD}_i \in 0^{n-r} [p]^r\} 
\quad
\mbox{By Proposition \ref{prop:1}}
\\
&=
\{ x, \ \hat{x} + \widehat{\vD_i} \in 0^{n-r} [p]^r\} 
\quad
\mbox{Re-indexing}
\\
&=
\{ x, \ \widehat{x} \in 0^{n-r} [p]^r\} 
\quad
\mbox{By the assumption that $\vD_i \leq \sigma$} \\
&=
\{ x, \ x \in [\sigma]\}  \\
&= 
\brs
\end{align}
It follows that:
$$
\cubet{\vzero}\cube{\vD} = 1
$$
Substituting in Equation \ref{eq:prod}
implies: 
$\big\| \cube{\vy} - \cubep{\vy+\vD} \big\|
=0$.

\item{\textbf{Item:} \ref{it:dist-far}}

\noindent
If $\|\vD\|_{M} \geq n \cdot\sigma+1$
there exists at least one coordinate $i\in [n]$ such that
$$
\widetilde{\vD_i} > 
\sigma
$$
in that case $\widehat{\vD_i} \notin 0^{n-r} [p]^r$,
and together with the assumption $\sigma = p^r$ we have:
$$
\forall x\in \brs
\quad
\hat{x} + \widehat{\vD_i} \notin 0^{n-r} [p]^r
$$
so
$$
\cubet{\vzero}\cube{\vD} = 0
$$
\end{proof}

We conclude from the lemma above that $\pcs{\widehat{\va}}$ is an eigenvector of $U_{\vt}$, for $\vt$ that is $\sigma$-close to a  word $\vs$ with eigenvalue $\omega_p^{-\hat{\va}\cdot \hat{\vs} }$.
\begin{lemma}\label{lem:ev}
Let $\pcs{\widehat{\va}}$ denote a PCS state with label $\widehat{\va} \in \F_p^{m \cdot k}$ and parameter $\sigma = p^r$ for integer $r<m$, and let $\vt \in \F_q^n$ such that
$$
\Delta_{M,q}(
\vt, 
\matA \vs)
\leq 
\sigma/n
$$
for some $\vs \in \F_q^k$.
Then
$$
U_{{\vt}} \pcs{\widehat{\va}} = \omega_p^{-\hat{\va} \cdot\hat{\vs}} \cdot 
\pcs{\widehat{\va}}.
$$
\end{lemma}

\begin{proof}

Since $\Delta_{M,q}(\vt, \matA \vs) \leq \sigma/n$, for $\sigma = p^r$, $r<m$ then by Proposition \ref{prop:gap} we can write:
$$
\vt = \matA \vs + \vD
$$
where $\| \vD\|_{M,q} \leq \sigma$.
Therefore
\begin{align}
U_{\vt} \pcs{\widehat{\va}}
&=
q^{-k/2}
\cdot
U_{\vt} 
\cdot
\sum_{\vc \in \F_q^k}
\omega_p^{\hat{\va} \hat{\vc}}
\cube{\matA \vc} \\
&=
q^{-k/2}
\cdot
U_{\vD}
\cdot
U_{\matA \vs}
\cdot
\sum_{\vc \in \F_q^k}
\omega_p^{\hat{\va} \hat{\vc}}
\cube{\matA \vc} \\
&=
q^{-k/2}
\cdot
U_{\vD}
\cdot
\sum_{\vc \in \F_q^k}
\omega_p^{\hat{\va} \hat{\vc}}
\cube{\matA \vc + \matA \vs} 
\mbox{ definition of shift over $\F_q$}
\\
&=
q^{-k/2}
\cdot
U_{\vD}
\cdot
\sum_{\vc \in \F_q^k}
\omega_p^{\hat{\va} \hat{\vc}}
\cube{\matA (\vc + \vs)} 
\mbox{ linearity over $\F_q$}
\\
&=
q^{-k/2}
\cdot
U_{\vD}
\cdot
\omega_p^{-\hat{\va} \hat{\vs}}
\sum_{\vc \in \F_q^k}
\omega_p^{\hat{\va} (\hat{\vc} + \hat{\vs})}
\cube{\matA (\vc + \vs)} \\
&=
q^{-k/2}
\cdot
U_{\vD}
\cdot
\omega_p^{-\hat{\va} \hat{\vs}}
\sum_{\vc \in \F_q^k}
\omega_p^{\hat{\va} \cdot \widehat{\vc + \vs}}
\cube{\matA (\vc + \vs)} 
\quad
\mbox{Proposition \ref{prop:1}}
\\
&=
\omega_p^{-\hat{\va} \hat{\vs}}
\cdot
q^{-k/2}
\cdot
U_{\vD}
\sum_{\vc \in \F_q^k}
\omega_p^{\hat{\va} \hat{\vc}}
\cube{\matA \vc} 
\quad
\mbox{Re-indexing $c+s \to c$}
\\
&=
\omega_p^{-\hat{\va} \hat{\vs}}
\cdot
q^{-k/2}
\sum_{\vc \in \F_q^k}
\omega_p^{\hat{\va} \hat{\vc}}
\cube{\matA \vc}
\quad
\mbox{Item \ref{it:dist-close} since $\|\vD\|_{M,q} \leq \sigma$} \\
&=
\omega_p^{-\hat{\va} \hat{\vs}}
\pcs{\widehat{\va}}
\end{align}

\end{proof}

We now show an efficient algorithm for sampling a PCS state $\pcs{\hat\va}$ for random $\hat\va$:
\begin{lemma}\label{lem:pcsgen}

\textbf{An efficient quantum PCS sampler}

\noindent
Let ${\cal C} = [n,k,d]$ be a code of $\F_q^n$ generated by matrix $\matA \in \F_q^{n \times k}$,
$q = p^m$,i.e.
$$
d = 
\min_{\vx\neq \vy, \vx, \vy \in {\cal }} \Delta_{M,q}(\vx, \vy)
$$
There exists a quantum algorithm that samples $\pcs{\hat{\va}}$ for $\hat{\va} \sim 
U\left[\F_p^{m \cdot k}\right]$ in time $\poly(n, \log(q))$, whenever $\sigma < d/n$.

\end{lemma}

\begin{proof}

Consider the following evolution according to the computational steps specified in each equation:
\begin{align}
\ket{\vzero}_1 \otimes \ket{\vzero}_2
&\to
\ket{\vzero}_1 
\otimes
\cube{\vzero}_2 
\quad
\mbox{QFT: } I \otimes {\cal F}_\sigma^n
\\
&\to
q^{-k/2}
\cdot
\sum_{\vc \in \F_q^k} \ket{\vc} 
\otimes
\cube{\vzero} 
\quad
\mbox{QFT: }{\cal F}_q^k \otimes I
\\
&\to
q^{-k/2}
\cdot
\sum_{\vc \in \F_q^k} \ket{\vc} 
\otimes
\cube{\matA\vc} 
\quad
\mbox{ controlled shift by } \matA \vc
\\
&\to
q^{-k/2}
\cdot
\sum_{\vc \in \F_q^k} \ket{\hat{\vc}} 
\otimes
\cube{\matA\vc} 
\quad
\mbox{Change rep.: }
\F_q^k \mbox{ to } \F_p^{m \cdot k}
\\
&\to
q^{-k/2}
\cdot
\sum_{\vc \in \F_q^k} 
\left(
q^{-k/2}
\cdot
\sum_{\hat{\va} \in \F_p^{mk}}
\omega_p^{\hat{\va} \hat{\vc}}
\ket{\hat{\va}}
\right) 
\otimes
\cube{\matA \vc} 
\quad
\mbox{QFT: }
{\cal F}_p^{m \cdot k}\otimes I
\\
&=
q^{-k}
\sum_{\hat{\va} \in \F_p^{mk}}
\ket{\hat{\va}}
\otimes
\left(
\sum_{\vc \in \F_q^k}
\omega_p^{\hat{\va} \hat{\vc}}
\cube{\matA \vc}
\right)
\end{align}
By definition we have $$
\forall \vc \in \F_q^k
\quad
\| \matA \vc \|_M 
\equiv
\Delta_{M,q}(\matA \vc, 0)
\geq
d > \sigma \cdot n.
$$
Then by Item \ref{it:dist-far} it follows that the set $\left\{ \cube{\matA \vc} \right\}_{c\in \F_q^k}$ forms an orthonormal set.  Hence
$$
\forall \hat{\va}\in \F_p^{mk} \quad
\P(\widehat{\va})
=
q^{-k}
$$
which is independent of $\widehat{\va}$, i.e. $\hat{a}$ is sampled uniformly from $\Z_p^{mk}$.
The running time of the procedure is determined by the complexity of the Fourier transform over $\F_p^{m k}$, which is at most
$$
\log(p) \cdot m \cdot k = \poly(n, \log(q)).
$$

\end{proof}

%

\section{An Algorithm for BNCP for Prime-Power Fields}

We now define the following quantum bounded-distance decoder:
We first define the algorithm in terms of $q = 2^m$ for simplicity of exposition, and later we'll generalize it to any $q = p^m$ for prime $p$:
\begin{algorithm}\label{alg:main}

\textbf{A Quantum Decoder for Finite Field BNCP}

\noindent
\textbf{Input:}
$(\matA \in \F_q^{n \times k}, \vt \in \F_q^n)$, $q = 2^m$, and parameter $\sigma>0$.

\begin{enumerate}

\item\label{it:sample}
Sample $T = k \cdot m$ quantum PCS states with parameter $\sigma$:
$$
\pcs{\widehat{\va_1}} \otimes \hdots \otimes \pcs{\widehat{\va_T}}
$$
\item
Let $\hat\matA$ denote the matrix whose columns are the labels of the sampled PCS states:
$$
\hat\matA = [\widehat{\va_1}, \hdots, \widehat{\va_T}]
$$
Assume w.l.o.g. that $\hat\matA$ is invertible over $\F_2$.
\item
Tensor with the uniform superposition
$$
q^{-k/2} \cdot \sum_{\vz \in \F_q^k} \ket{\widehat{\vz}}
$$ 
\item
Apply $\hat\matA^{-1}$ to the register:
$$
q^{-k/2} \cdot \sum_{\vz \in \F_q^k} \ket{\hat\matA^{-1} \cdot \widehat{\vz}}
$$ 
\item
Apply a controlled-shift operation where bit $j \in [T]$ of $\hat\matA^{-1} \widehat{\vz} \in \F_2^T$ controls whether or not we apply $U_{\vt}$ to the $j$-th PCS state:
$$
q^{-k/2} \cdot 
\sum_{\vz \in \F_q^k} 
\ket{\hat\matA^{-1} \cdot \widehat{\vz}}
\otimes
U_{\vt}^{(\hat\matA^{-1} \widehat{\vz})_1}\pcs{\widehat{\va_1}}
\otimes
\hdots
\otimes
U_{\vt}^{(\hat\matA^{-1} \widehat{\vz})_T}\pcs{\widehat{\va_T}}
$$
\item\label{it:5}
Apply $\hat\matA$ to the first register:
$$
q^{-k/2} \cdot \sum_{\vz \in \F_q^k} \ket{\widehat{\vz}}
\otimes
U_{\vt}^{(\hat\matA^{-1} \widehat{\vz})_1}\pcs{\widehat{\va_1}}
\otimes
\hdots
\otimes
U_{\vt}^{(\hat\matA^{-1} \widehat{\vz})_T}\pcs{\widehat{\va_T}}
$$
\item
Apply the $\F_2^T$ quantum Fourier transform on the first register, and measure in the standard basis.
Denote as output ${\cal O}$.
\end{enumerate}

\end{algorithm}

Using this algorithm we solve an instance of $\eps$-BNCP to factor $\eps = 1/(2n^2)$.
\begin{theorem}\label{thm:main}
Let ${\cal C} = [n,k,d]$ be an error correcting code over $\F_q^n$ for $q=p^m$, $p=2$.
Let $(\matA \in \F_q^{n\times k}, \vt\in \F_q^n)$ 
be an instance of $\eps$-BNCP where
$$
\Delta(
\vt , \matA \vs )
\leq d/(2n^2)
$$
for some $\vs \in \F_q^k$.
Then upon input $(\matA, \vt)$ and  parameter $\sigma = 2^r, r<m$ that satisfies:
$$
(\ast)
\quad
d/(2n) 
\leq
\sigma
<
d/n
$$   
Algorithm \ref{alg:main} runs in expected time $\poly(n, \log(q))$ and returns
an outcome ${\cal O} = \hat{\vs}$.
\end{theorem}
We note that the theorem above assumes a-priori knowledge of $d$.  
This is reasonable in the error-correction setting, but in the computational theory of lattices knowledge of the minimal distance amounts to an oracle to the GapSVP problem which is also known to be a hard problem.
However, by initializing $\sigma = 1$ and executing the algorithm on sequential doubling of the parameter there will be at least one iteration such that $\sigma = 2^r$ satisfies the condition $(\ast)$.
Since the correct answer can be easily checked this essentially removes the need to know $d$ in advance.

\begin{proof}

Assume for now that $\hat\matA$ is invertible over $\F_2$ and
consider the output of step \ref{it:5}
$$
\ket{\psi}
=
q^{-k/2} \cdot \sum_{\vz \in \F_q^k} \ket{\widehat{\vz}}
\otimes
U_{\vt}^{(\hat\matA^{-1} \widehat{\vz})_1}\pcs{\widehat{\va_1}}
\otimes
\hdots
\otimes
U_{\vt}^{(\hat\matA^{-1} \widehat{\vz})_T}\pcs{\widehat{\va_T}}
$$
By our choice of parameters we have:
$$
\sigma \geq n \cdot \Delta_{M,q}(\vt, \matA  \vs).
$$
Since in addition $\sigma = 2^r$, $r<m$,
we can invoke Lemma \ref{lem:ev} which implies:
$$
\forall \widehat{\va} \in \F_2^T
\quad
U_{\vt} \pcs{\widehat{\va}}
=
(-1)^{\hat{\va} \hat{\vs}}
\cdot
\pcs{\widehat{\va}}.
$$
Observe that:
$$
U_{\vt}^1 = U_{\vt}
\quad
U_{\vt}^0 = I.
$$
Hence each PCS state $\pcs{\widehat{\va_i}}$ above is multiplied by a phase $(-1)^{\widehat{\va_i} \hat{\vs}}$ if $(\hat\matA^{-1} \widehat{\vz})_i = 1$ and by phase $1$ if $(\hat\matA^{-1} \widehat{\vz})_i = 0$:
\begin{align}
\ket{\psi}
&=
q^{-k/2} 
\cdot 
\sum_{\vz \in \F_q^k} 
\ket{\widehat{\vz}}
\otimes
(-1)^{(\hat\matA^{-1} \widehat{\vz})_1 \cdot 
\widehat{\va_1} \cdot \hat{\vs}}
\pcs{\widehat{\va_1}}
\otimes
\hdots
\otimes
(-1)^{(\hat\matA^{-1} \widehat{\vz})_T \cdot 
\widehat{\va_T} \cdot \hat{\vs}}
\pcs{\widehat{\va_T}} \\
&=
q^{-k/2} \cdot 
\sum_{\vz \in \F_q^k} 
\ket{\widehat{\vz}}
\otimes
(-1)^{\hat{\vs} \cdot \hat\matA \cdot \hat\matA^{-1} \cdot \hat{\vz}}
\pcs{\widehat{\va_1}}
\otimes
\hdots
\otimes
\pcs{\widehat{\va_T}} \\
&=
\left(
q^{-k/2} \cdot 
\sum_{\vz \in \F_q^k} 
(-1)^{\hat{\vs} \cdot \hat{\vz}}
\ket{\hat{\vz}}
\right)
\otimes
\pcs{\widehat{\va_1}}
\otimes
\hdots
\otimes
\pcs{\widehat{\va_T}}
\end{align}
In this case measuring register 1 in the $\F_2^T$ Fourier basis results in outcome $\hat{\vs}$ with probability $1$.

\noindent
\textbf{Running time:}
Since 
$$
\sigma < d/n
$$
then by Lemma \ref{lem:pcsgen} we can sample PCS states $\pcs{\widehat{\va}}$ such that $\widehat{\va} \sim U[\F_2^T]$ in time $\poly(n, \log(q))$.
By independence of sampling this implies that the entries of $\hat\matA$ are i.i.d. uniform on $\F_2$.
By Lemma \ref{lem:inv} this implies that
$$
\P( \hat\matA \mbox{ is invertible})
\geq \prod_{k=1}^{\infty} (1 - 2^{-k})
- e^{-c  T}
\geq 1/10
$$
It follows that after $O(1)$ iterations of Step \ref{it:sample} the matrix $\hat\matA$ is invertible.
The rest of the computational steps: namely the Quantum Fourier Transform, the controlled shift operation, and multiplication by $\hat\matA, \hat\matA^{-1}$ all take time at most $\poly(n, \log(q))$.

\end{proof}

\subsection{Generalization to Arbitrary Characteristic}

In the previous section we have shown an algorithm to solve BNCP on fields $\F_q$ of characteristic $2$, namely $q = 2^m$ for some integer $m$.
In this section we'll generalize this algorithm to arbitrary characteristic: $q = p^m$ for prime $p$.

We consider again Algorithm \ref{alg:main} previously stated for $p=2$.  For general $p$ we require in Step \ref{it:sample} that $\hat\matA$ is invertible over $\F_p$, and in Step \ref{it:shift1} we consider operators of the form $U_{\vt}^\ell$ where now $\ell$ can assume any number in $\F_p$ (instead of a binary value)
and $U_{\vt}^\ell$ is then interpreted as taking the $\ell$-th power of $U_{\vt}$ where:
$$
U_{\vt}^\ell = 
\underbrace{
U_{\vt}
\cdot
\hdots
\cdot
U_{\vt}}_{\ell \mbox{ times}}
$$ 
\\

\noindent
We now restate Theorem \ref{thm:main}
for prime-power fields $q = p^m$.
We note that the distance to the lattice for which the theorem holds is now decreased by a factor of $p$, i.e. we can solve the problem when the distance $\Delta(\vt, {\cal C})$ is at most $d/(p \cdot n^2)$.  This extra condition is set in order to allow the existence of a value $\sigma = p^r, r<m$ that is at least $\Delta(\vt, {\cal C})$ and at most $d/n$.
As before, this condition can be omitted by making a numerical assumption on the distance from ${\cal C}$.

\begin{theorem}\label{thm:main2}

\noindent
Let ${\cal C} = [n,k,d]$ be an error correcting code over $\F_q^n$ for $q=p^m$ for prime $p$.
Let $(\matA \in \F_q^{n\times k}, \vt\in \F_q^n)$ 
be an instance of $\eps$-BNCP where
$$
\Delta_{M,q}( \vt , \matA \vs) \leq d / (pn^2)
$$ 
for some $\vs \in \F_q^k$.
Then upon input $(\matA, \vt)$ and  parameter $\sigma = p^r, r<m$ that satisfies:
$$
(\ast)
\quad
d/(pn)
<
\sigma
<
d/n
$$   
Algorithm \ref{alg:main} runs in expected time $\poly(n,p, \log(q))$ and returns
results an outcome ${\cal O} = \hat{\vs}$.
\end{theorem}

\begin{proof}

Assume for now that $\hat\matA$ is invertible over $\F_p$ and
consider the output of step \ref{it:5}
$$
\ket{\psi}
=
q^{-k/2} \cdot \sum_{\vz \in \F_q^k} \ket{\widehat{\vz}}
\otimes
U_{\vt}^{(\hat\matA^{-1} \widehat{\vz})_1}\pcs{\widehat{\va_1}}
\otimes
\hdots
\otimes
U_{\vt}^{(\hat\matA^{-1} \widehat{\vz})_T}\pcs{\widehat{\va_T}}
$$
By our choice of parameters we have:
$$
\sigma > n \cdot \Delta(\vt,\matA \vs)
$$
and $\sigma = p^r$ for $r<m$.
Thus we can invoke
Lemma \ref{lem:ev} which implies:
$$
\forall \widehat{\va} \in \F_p^T, 
\quad
U_{\vt} \pcs{\widehat{\va}}
=
\omega_p^{-\hat{\va} \hat{\vs}}
\cdot
\pcs{\widehat{\va}}.
$$
Therefore
$$
\forall \ell \in \F_p
\quad
U_{\vt}^\ell
=
U_{\vt} U_{\vt} \cdot \hdots \cdot U_{\vt} \pcs{\widehat{\va}}
=
\omega_p^{-\ell \cdot \hat{\va} \hat{\vs}}
\cdot \pcs{\widehat{\va}}
$$
Hence each PCS state $\pcs{\va_i}$ above is multiplied by a phase 
$\omega_p^{- \ell \cdot \hat{\va}_i \hat{\vs}}$ where 
$\ell = (\hat\matA^{-1} \widehat{\vz})_i\in \F_p$:
\begin{align}
\ket{\psi}
&=
q^{-k/2} 
\cdot 
\sum_{\vz \in \F_q^k} 
\ket{\widehat{\vz}}
\otimes
\omega_p^{-(\hat\matA^{-1} \widehat{\vz})_1 \cdot \widehat{\va_1} \cdot \hat{\vs}}
\pcs{\widehat{\va_1}}
\otimes
\hdots
\otimes
\omega_p^{-(\hat\matA^{-1} \widehat{\vz})_T \cdot \widehat{\va_T} \cdot \hat{\vs}}
\pcs{\widehat{\va_T}} \\
&=
q^{-k/2} \cdot 
\sum_{\vz \in \F_q^k} 
\ket{\widehat{\vz}}
\otimes
\omega_p^{-\hat{\vs} \cdot \hat\matA \cdot \hat\matA^{-1} \cdot \hat{\vz}}
\pcs{\widehat{\va_1}}
\otimes
\hdots
\otimes
\pcs{\widehat{\va_T}} \\
&=
\left(
q^{-k/2} \cdot 
\sum_{\vz \in \F_q^k} 
\omega_p^{-\hat{\vs} \cdot \hat{\vz}}
\ket{\hat{\vz}}
\right)
\otimes
\pcs{\widehat{\va_1}}
\otimes
\hdots
\otimes
\pcs{\widehat{\va_T}}
\end{align}
In this case measuring register 1 in the $\F_p^T$ Fourier basis results in outcome $-\hat{\vs}$ with probability $1$. Taking the negation of the answer yields ${\cal O} = \hat{\vs}$.

\noindent
\textbf{Running time:}
Since 
$$
\sigma < d/n
$$
then by Lemma \ref{lem:pcsgen} we can sample PCS states $\pcs{\widehat{\va}}$ such that $\widehat{\va} \sim U[\F_p^T]$ in time $\poly(n, \log(q))$.
So by independence of sampling this implies that the entries of $\hat\matA$ are i.i.d. uniform on $\F_p$.
By Lemma \ref{lem:inv} this implies that
$$
\P( \hat\matA \mbox{ is invertible})
\geq \prod_{k=1}^{\infty} (1 - p^{-k})
- e^{-c  T}
\geq 
1 - 
\sum_{k=1}^{\infty}
p^{-k}
-
e^{-c T}
=
1 - \frac{1/p}{1-1/p}
 - e^{-c T}
$$
$$
=
\frac{1 - 2/p}{1-1/p}
- e^{-cT}
\geq 
1/4
$$
where the last inequality follows from assuming $p\geq 3$ and sufficiently large $T$.
It follows that after $O(1)$ iterations of Step \ref{it:sample} the matrix $\hat\matA$ is invertible.
The rest of the computational steps: namely the Quantum Fourier Transform, the controlled shift operation, and multiplication by $\hat\matA, \hat\matA^{-1}$ all take time at most $\poly(n, p, \log(q))$, where the extra factor of $p$ comes from the fact that $U_{\vt}^p$ is implemented as $p$ sequential applications of $U_{\vt}$.

\end{proof}

\section{Acknowledgements}
The author thanks L{\'e}o Ducas, Saeed Mehraban, Peter Shor,  Nicolas Sendrier, and an anonymous reviewer for their useful comments and suggestions.

\printbibliography

@inproceedings{Pei09,
  author    = {Chris Peikert},
  title     = {Public-key cryptosystems from the worst-case shortest vector problem:
               extended abstract},
  booktitle = {Proceedings of the 41st Annual {ACM} Symposium on Theory of Computing,
               {STOC} 2009, Bethesda, MD, USA, May 31 - June 2, 2009},
  pages     = {333--342},
  publisher = {{ACM}},
  year      = {2009},
  url       = {https://doi.org/10.1145/1536414.1536461},
  doi       = {10.1145/1536414.1536461},
  bibsource = {dblp computer science bibliography, https://dblp.org}
}

@article{EH22,
  author    = {Lior Eldar and
               Sean Hallgren},
  title     = {An efficient quantum algorithm for lattice problems achieving subexponential
               approximation factor},
  journal   = {CoRR},
  volume    = {abs/2201.13450},
  year      = {2022},
  url       = {https://arxiv.org/abs/2201.13450},
  eprinttype = {arXiv},
  eprint    = {2201.13450},
  bibsource = {dblp computer science bibliography, https://dblp.org}
}

@inproceedings{BLPRS13,
author = {Brakerski, Zvika and Langlois, Adeline and Peikert, Chris and Regev, Oded and Stehl\'{e}, Damien},
title = {Classical Hardness of Learning with Errors},
year = {2013},
isbn = {9781450320290},
doi = {10.1145/2488608.2488680},
booktitle = {STOC'13},
pages = {575-584},
numpages = {10},
location = {Palo Alto, California, USA},
}

@article{Reg09,
	Address = {New York, NY, USA},
	Author = {Regev, Oded},
	Issn = {0004-5411},
	Journal = {J. ACM},
	Number = {6},
	Pages = {1--40},
	Publisher = {ACM},
	Title = {On lattices, learning with errors, random linear codes, and cryptography},
	Volume = {56},
	Year = {2009}}

@ARTICLE{BMT78,  author={Berlekamp, E. and McEliece, R. and van Tilborg, H.},  journal={IEEE Transactions on Information Theory},   title={On the inherent intractability of certain coding problems (Corresp.)},   year={1978},  volume={24},  number={3},  pages={384-386},  doi={10.1109/TIT.1978.1055873}}

@INPROCEEDINGS{Reg03,  author={Regev, O.},  booktitle={18th IEEE Annual Conference on Computational Complexity, 2003. Proceedings.},   title={Improved inapproximability of lattice and coding problems with preprocessing},   year={2003},  volume={},  number={},  pages={363-370},  doi={10.1109/CCC.2003.1214435}}

@article{ABSS97,
title = {The Hardness of Approximate Optima in Lattices, Codes, and Systems of Linear Equations},
journal = {Journal of Computer and System Sciences},
volume = {54},
number = {2},
pages = {317-331},
year = {1997},
issn = {0022-0000},
doi = {https://doi.org/10.1006/jcss.1997.1472},
url = {https://www.sciencedirect.com/science/article/pii/S0022000097914720},
author = {Sanjeev Arora and László Babai and Jacques Stern and Z Sweedyk}
}

@InProceedings{CS98,
author="Canteaut, Anne
and Sendrier, Nicolas",
editor="Ohta, Kazuo
and Pei, Dingyi",
title="Cryptanalysis of the Original McEliece Cryptosystem",
booktitle="Advances in Cryptology --- ASIACRYPT'98",
year="1998",
publisher="Springer Berlin Heidelberg",
address="Berlin, Heidelberg",
pages="187--199"
}

@InProceedings{NIST20,
url=
"https://csrc.nist.gov/News/2020/pqc-third-round-candidate-announcement"}

@InProceedings{LM09,
author="Lyubashevsky, Vadim
and Micciancio, Daniele",
editor="Halevi, Shai",
title="On Bounded Distance Decoding, Unique Shortest Vectors, and the Minimum Distance Problem",
booktitle="Advances in Cryptology - CRYPTO 2009",
year="2009",
publisher="Springer Berlin Heidelberg",
address="Berlin, Heidelberg",
pages="577--594"
}

@InProceedings{PS10,
author="Peters, Christiane",
editor="Sendrier, Nicolas",
title="Information-Set Decoding for Linear Codes over Fq",
booktitle="Post-Quantum Cryptography",
year="2010",
publisher="Springer Berlin Heidelberg",
address="Berlin, Heidelberg",
pages="81--94"
}

@article{M10,
  title={Singularity of Random Matrices over Finite Fields},
  author={Kenneth Maples},
  journal={arXiv: Combinatorics},
  year={2010}
}

@article{DKRS03,
author = {Dinur, I. and Kindler, Guy and Raz, R. and Safra, S.},
year = {2003},
month = {04},
pages = {205-243},
title = {Approximating CVP to Within Almost-Polynomial Factors is NP-Hard},
volume = {23},
journal = {Combinatorica},
doi = {10.1007/s00493-003-0019-y}
}

@article{mceliece,
title={Classic McEliece: conservative code-based cryptography},
author= {Bernstein, Daniel J. and Chou, Tung and Lange, Tanja and von Maurich, Ingo and Misoczki, Rafael and Niederhagen, Ruben and 
Persichetti, Edoardo and Peters, Christiane  and Schwabe, Peter and Sendrier, Nicolas and Szefer, Jakub and Wang, Wen},
journal = {PQCRYPTO Mini-School and Workshop},
year={2018}
}

@article{HW19,
author = {Horlemann-Trautmann, Anna-Lena and Weger, Violetta},
year = {2019},
month = {01},
pages = {},
title = {Information set decoding in the Lee metric with applications to cryptography},
volume = {15},
journal = {Advances in Mathematics of Communications},
doi = {10.3934/amc.2020089}}

@inproceedings{CDE21,
  author    = {Andr{\'{e}} Chailloux and
               Thomas Debris{-}Alazard and
               Simona Etinski},
  editor    = {Jung Hee Cheon and
               Jean{-}Pierre Tillich},
  title     = {Classical and Quantum Algorithms for Generic Syndrome Decoding Problems
               and Applications to the Lee Metric},
  booktitle = {Post-Quantum Cryptography - 12th International Workshop, PQCrypto
               2021, Daejeon, South Korea, July 20-22, 2021, Proceedings},
  series    = {Lecture Notes in Computer Science},
  volume    = {12841},
  pages     = {44--62},
  publisher = {Springer},
  year      = {2021},
  url       = {https://doi.org/10.1007/978-3-030-81293-5\_3},
  doi       = {10.1007/978-3-030-81293-5\_3},
  bibsource = {dblp computer science bibliography, https://dblp.org}
}

\appendix
\section{Proof of Technical Lemmas}

\subsection{Proof of Lemma \ref{lem:random}}

\begin{proof}
Consider $\va_1,\hdots, \va_k$ random vectors in $\F_q^n$ that generate ${\cal C}$, and a vector of coefficients $\vx = (x_1,\hdots, x_k) \in \F_q^k$. 
We have
\begin{align}
\P
\left(
\exists \vc\in {\cal C},\ \vc \neq 0
\ \|\vc\|_{\rm M} \leq n \cdot q^{1-k/n}/2
\right) 
&\leq
\P_{\va_1,\hdots \va_k}
\left(
\exists \vx \in \F_q^k,
\vx \neq 0, 
\quad
\left\|\sum_{i=1}^k \va_i x_i \right\|_{\rm M} \leq n \cdot q^{1-k/n}/2
\right) \\
&\leq
q^k
\cdot
\P_{\va_i \sim U[\F_q^n], \vx \neq 0}
\left(
\left\|\sum_{i=1}^k \va_i x_i \right\|_{\rm M} \leq n \cdot q^{1-k/n}/2
\right)  
\end{align}
Considering the above, for any nonzero $\vx\in \F_q^k$ each coordinate $z_j\in \F_q, j\in [n]$ of $\vz = \sum_i \va_i x_i$ is a uniformly
random variable on $\F_q$ that is independent of all other coordinates.
We have
$$
\P( |z_j| \leq q^{1-k/n} )
\leq q^{-k/n}. 
$$
Thus, by Chernoff:
$$
\P
\left(
\sum_i |z_i|
\leq
n \cdot q^{1-k/n}/2
\right)
\leq
2^{-n \cdot (q^{-k/n})^2/16}
\leq
2^{-n / 200}
$$
where the last inequality follows from $n \geq k \log(q)$.
Applying the contrapositive of Proposition \ref{prop:gap}
we conclude that
$$
\P
\left(
\min_{\vx\neq \vy, \vx, \vy \in {\cal C}}
\Delta_{M,q}(\vx, \vy) < q^{1-k/n}/2
\right) < 2^{-n/200}
$$
\end{proof}

\subsection{Proof of Theorem \ref{thm:hardness}}

Consider the problem of $c$-approximate set-cover: a ground-set ${\cal U}$, and a collection of subsets $S_1,\hdots, S_m$. A cover is a sub-collection of the $S_i$'s whose union is ${\cal U}$.
The cover is exact if the sets in the cover are disjoint. 
The size of a cover is the number of sets that comprise it.

The construction of \cite{ABSS97} defines $m+1$ vectors $b_0,\hdots, b_m \in \F_2^r$, $r = L |{\cal U}| + m$, $L = c K$ as follows:
for each set $S_i$ we define a vector $b_i$ on $L \cdot |U| + m$.
The first $L |{\cal U}|$ coordinates are considered as $|{\cal U}|$ tuples of $L$ coordinates each, where each tuple corresponds to an element of ${\cal U}$.
$b_i$ is zero except for the $L \cdot |S_i|$ coordinates corresponding to the characteristic vector of $S_i$.  The last $m$ coordinates are zero except at the $i$-th position which is $1$.
The vector $b_0$ is the all-ones vector on the first $L |{\cal U}|$ coordinates, and $0$ on the last $m$.
We now claim similarly to \cite{ABSS97}:
\begin{lemma}
Let $q = p^m$ and suppose that $c>p$.
Let $L = c \cdot K$.
Define:
$$
{\rm OPT} = \min_{\alpha}
\dman
\left( b_0, \sum_i \alpha_i b_i 
\right)
$$
and let ${\cal C}$ denote the linear span of the vector $b_1,\hdots, b_m$ over $\F_q$.
If there exists an exact cover of size $K$ then
\be\label{eq:yes}
\dman(b_0, {\cal C}) \leq p \cdot K
\ee
and if any cover is of size at most $c \cdot K$ then
\be\label{eq:no}
\dman(b_0, {\cal C}) 
\geq
c  \cdot K 
\ee
\end{lemma}

\begin{proof}
Let $p$ denote the characteristic of $\F_q$, i.e. $q = p^m$ for some integer $m>0$.
If there exists an exact cover $S_{i_1}, \hdots, S_{i_K}$ then choosing 
$\alpha_{i_j}= p-1$ for all $j\in [K]$ and $0$ otherwise has that
$b_0 + \sum_i \alpha_i b_i$ is equal to $0$ on the first $|{\cal U}| L$ bits.  On the last $m$ bits the Manhattan weight of
$b_0 + \sum_i \alpha_i b_i$
is precisely $\sum_i \alpha_i = (p-1) \cdot K$.
Hence
$$
\dman(b_0, {\cal C}) 
\leq
(p-1) \cdot K
$$

Suppose now that any cover has size at least $L = c \cdot K$.
Let $\alpha = (\alpha_1,\hdots, \alpha_r)$ denote an assignment vector $\alpha_i \in \F_q$.
First, suppose that $\sum_i \alpha_i b_i$ has non-zero coordinates on each of the first $|{\cal U}|$ tuples of $L$ bits.
Then each tuple is "covered" by at least one vector $b_i$, that corresponds to set $S_i$, and $b_i$ is multiplied by a non-zero coefficient $\alpha_i$. Thus $\sum_i \alpha_i \geq c \cdot K = L$ 
this is manifested in the last $m$ coordinates, implying
$$
\dman(b_0, \sum_i \alpha_i b_i) \geq L 
$$
On the other hand, if not all tuples are covered, i.e. there is at least one $L$-tuple that is all zeros, then the Manhattan distance on the first $|{\cal U}| L$ coordinates is at least $L$, implying
$$
\dman(b_0, \sum_i \alpha_i b_i) \geq L 
$$

\end{proof}
The proof of Theorem \ref{thm:hardness} follows by applying the lemma in conjunction with the fact that there exists a constant $c>0$ such that it is NP-hard to approximate exact set-cover to factor at most $c$.

\end{document}